\documentclass[10pt]{amsart}
\usepackage{amsfonts}
\usepackage{amssymb}
\usepackage{mathrsfs}
\usepackage{amsmath,amsthm}
\usepackage{amscd}
\usepackage{graphics}
\usepackage[,all]{xy}
\usepackage{graphicx}
\usepackage{amsmath,amssymb}
\usepackage{cancel}
\usepackage{pdfsync}
\newcommand{\rd}{\mathrm{d}}
\newcommand{\beq}{\begin{equation}}
\newcommand{\eeq}{\end{equation}}
\newcommand{\bea}{\begin{eqnarray}}
\newcommand{\eea}{\end{eqnarray}}

\newcommand\noi{\noindent}

\newcommand{\bk}{\begin{cases}}
\newcommand{\ek}{\end{cases}}

\newtheorem{theorem}{Theorem}

\newtheorem{definition}[theorem]{Definition}
\newtheorem{example}[theorem]{Example}

\newtheorem{lemma}[theorem]{Lemma}

\newtheorem{proposition}[theorem]{Proposition}
\newtheorem{remark}[theorem]{Remark}

\begin{document}
\author{Giorgio Tondo}
\address{Dipartimento di Matematica e Geoscienze, Universit\`a  degli Studi di Trieste,
piaz.le Europa 1, I--34127 Trieste, Italy.}
\email{tondo@units.it}

\title[Haantjes Algebras of the Lagrange top]{Haantjes Algebras of the Lagrange top}

\begin{abstract}
A  symplectic--Haantjes manifold and a Poisson--Haantjes manifold for the Lagrange top are studied and a set of Darboux-Haantjes coordinates are computed. Such  coordinates are  separation variables for the  associated Hamilton-Jacobi equation.
\end{abstract}
 \maketitle
\tableofcontents

\section{Introduction}
In this article, we propose a   new  theoretical framework for dealing with integrable  mechanical systems, and we illustrate  it on  the classical
example of the Lagrange top. Precisely, we discuss such Hamiltonian system  in the framework of  the theory of Haantjes algebras, very recently introduced in \cite{TT2017}. To this aim, the new geometric notion of Poisson--Haantjes ($P\mathcal{H}$) manifold  is proposed, as a natural extension of the  notion  of symplectic-Haantjes ($\omega \mathcal{H}$) manifold, already introduced in \cite{TT2016prepr}. The main idea is to replace the symplectic structure of a 
$\omega \mathcal{H}$ manifold with a Poisson (not invertible) bivector that fulfills a suitable algebraic compatibility condition with a Haantjes algebra of operators on the tangent bundle of the manifold.  Besides, the dynamical notion of Magri-Haantjes chain \cite{TT2016prepr}  is generalized, to be adapted to  the novel geometrical notion of $P\mathcal{H}$ manifold. These new  structures can be used to describe bi-Hamiltonian chains of vector fields for Gelfand-Zakarevich (GZ) manifolds  \cite{GZ}. When such structures can be reduced to the symplectic leaves of the Poisson bivector, a $\omega \mathcal{H}$ structure is obtained and a set of \textit{Darboux-Haantjes coordinates} can be computed. Such coordinates are   separation variables for the Hamilton-Jacobi equation of the Hamiltonian systems belonging to a Magri--Haantjes chain. 
Here, we detail the Lagrange top, that is a GZ system of corank $2$, whilst the discussion about the stationary flows of the KdV hierarchy (a GZ system of corank $1$) will appear elsewhere.  

The notion of $P\mathcal{H}$ structures, is inspired, from one side, by the theory of Poisson--Nijenhuis structures \cite{MM, KM}, from the other side by the notion of  Haantjes manifolds \cite{MFrob,MGall13,MGWDVV,MGall15}. In our opinion, the new theory provides us a very flexible and unifying theoretical framework for dealing with the integrability and
separability properties of Hamiltonian systems, and represents a formulation 
that completes  the one offered by the Poisson--Nijenhuis geometry.

\par
The paper is organized as follows. After a review, in Section 2, of the main algebraic structures needed in this work, we recall in Section 3  the concept of Haantjes algebras. In Section 4, we present the new notion of Poisson--Haantjes manifolds and of the related Magri--Haantjes chains. In Section 5, we apply the theory to the real and complex Lagrange top.

\section{Nijenhuis  and Haantjes torsion}
\label{sec:1}
 The natural frames $(\frac{\partial}{\partial x_1},\ldots,\frac{\partial}{\partial x_n})$ of  local charts $(x_1,\ldots,x_n)$ in a differentiable manifold, being obviously integrable, can be characterized in a \textit{tensorial manner} as  eigen-distributions of a suitable class of  $(1,1)$ tensor fields, i.e. the ones with vanishing  Haantjes tensor. In this section, we review some basic  results concerning the theory of such tensors. For a more complete treatment, see
the original papers \cite{Haa,Nij},  the related ones \cite{Nij2,FN}, and the nice recent review \cite{K17}.

Let $M$ be a differentiable manifold and $\boldsymbol{L}:TM\rightarrow TM$ be a $(1,1)$ tensor field, i.e., a field of linear operators on the tangent space at each point of $M$.
\begin{definition}
The
 \textit{Nijenhuis torsion} of $\boldsymbol{L}$ is the skew-symmetric  $(1,2)$ tensor field defined by
\begin{equation} \label{eq:Ntorsion}
\mathcal{T}_ {\boldsymbol{L}} (X,Y):=\boldsymbol{L}^2[X,Y] +[\boldsymbol{L}X,\boldsymbol{L}Y]-\boldsymbol{L}\Big([X,\boldsymbol{L}Y]+[\boldsymbol{L}X,Y]\Big),
\end{equation}
where $X,Y \in TM$ and $[ \ , \ ]$ denotes the commutator of two vector fields.
\end{definition}
In local coordinates $\boldsymbol{x}=(x_1,\ldots, x_n),$ the Nijenhuis torsion can be written in the form
\begin{equation}\label{eq:NtorsionLocal}
(\mathcal{T}_{\boldsymbol{L}})^i_{jk}=\sum_{\alpha=1}^n\bigg(\frac{\partial {\boldsymbol{L}}^i_k} {\partial x^\alpha} {\boldsymbol{L}}^\alpha_j -\frac{\partial {\boldsymbol{L}}^i_j} {\partial x^\alpha} {\boldsymbol{L}}^\alpha_k+\Big(\frac{\partial {\boldsymbol{L}}^\alpha_j} {\partial x^k} -\frac{\partial {\boldsymbol{L}}^\alpha_k} {\partial x^j}\Big ) {\boldsymbol{L}}^i_\alpha \bigg)\ ,
\end{equation}
amounting to $n^2(n-1)/2$ independent components. 
\begin{definition}
\noi The \textit{Haantjes torsion} associated with $\boldsymbol{L}$ is the $(1, 2)$ tensor field defined by
\begin{equation} \label{eq:Haan}
\mathcal{H}_{\boldsymbol{L}}(X,Y):=\boldsymbol{L}^2\mathcal{T}_{\boldsymbol{L}}(X,Y)+\mathcal{T}_{\boldsymbol{L}}(\boldsymbol{L}X,\boldsymbol{L}Y)-\boldsymbol{L}\Big(\mathcal{T}_{\boldsymbol{L}}(X,\boldsymbol{L}Y)+\mathcal{T}_{\boldsymbol{L}}(\boldsymbol{L}X,Y)\Big).
\end{equation}
\end{definition}
The skew-symmetry of the Nijenhuis torsion implies that the Haantjes tensor is also skew-symmetric.
Its local expression is

\begin{equation}\label{eq:HaanLocal}
(\mathcal{H}_{\boldsymbol{L}})^i_{jk}=  \sum_{\alpha,\beta=1}^n\bigg(
\boldsymbol{L}^i_\alpha \boldsymbol{L}^\alpha_\beta(\mathcal{T}_{\boldsymbol{L}})^\beta_{jk}  +
(\mathcal{T}_{\boldsymbol{L}})^i_{\alpha \beta}\boldsymbol{L}^\alpha_j \boldsymbol{L}^\beta_k-
\boldsymbol{L}^i_\alpha\Big( (\mathcal{T}_{\boldsymbol{L}})^\alpha_{\beta k} \boldsymbol{L}^\beta_j+
 (\mathcal{T}_{\boldsymbol{L}})^\alpha_{j \beta } \boldsymbol{L}^\beta_k \Big)
 \bigg) \ .
\end{equation}
We shall  consider a case, in which the computation of the  Haantjes torsion will be particularly simple \cite{TT2016prepr}.
\begin{proposition}\label{pr:Hd}
Let $\boldsymbol{L}$ be a smooth field of operators. If there exists a local  chart $\{U, (x_1,\ldots, x_n)\}$ where $\boldsymbol{L}$ takes the diagonal form 

\begin{equation}
\label{eq:Ldiagonal}
\boldsymbol{L}(\boldsymbol{x})=\sum_{i=1}^n l_i(\boldsymbol{x}) \frac{\partial}{\partial x_i} \otimes \rd x_i \ ,
\end{equation}
 then the Haantjes tensor of $\boldsymbol{L}$ vanishes.
\end{proposition}

Due to the relevance of the  Haantjes (Nijenhuis) vanishing condition,
we propose the following definition.
\begin{definition}
A Haantjes (Nijenhuis)   operator is a field of  operators whose  Haantjes (Nijenhuis) tensor identically vanishes.
\end{definition}

It has been proved in the following proposition   that a single Haantjes operator   generates an algebra of Haantjes operators over the ring of smooth functions on $M$.  This is not the case for a Nijenhuis operator  $\boldsymbol{N}$ since a polynomial in $\boldsymbol{N}$ with coefficients $a_j \in C^\infty(M)$,  is not necessarily a Nijenhuis operator.
\par
\begin{proposition} \label{pr:Lpowers}  \cite{BogI}.
Let $\boldsymbol{L}$ be a Haantjes operator in $M$. Then for any  polynomial in $\boldsymbol{L}$ with coefficients $a_{j}\in C^\infty(M)$, the associated Haantjes tensor vanishes, i.e.
\begin{equation}
\mathcal{H}_{\boldsymbol{L}}(X,Y)= 0 \ \Longrightarrow \
\mathcal{H}_{(\sum_j a_{j} (\boldsymbol{x}) \boldsymbol{L}^j)}(X,Y)= 0.
\end{equation}
\end{proposition}
\begin{proof}
See Corollary 3.3, p. 1136 of \cite{BogI}.
\end{proof}

\section{Haantjes algebras }
In this section we recall the  notion of \textit{Haantjes algebra} and the class of  \textit{cyclic} Haantjes algebras,
very recently introduced in \cite{TT2017}.

\begin{definition}\label{def:HM}
A Haantjes algebra of rank $m$ is a pair    $(M, \mathscr{H})$ which satisfies the following  conditions:
\begin{itemize}
\item
$M$ is a differentiable manifold of dimension $n$;
\item
$ \mathscr{H}$ is a set of Haantjes  operators $\boldsymbol{K}_i:TM\rightarrow TM$ that generates
\begin{itemize}
\item
 a free module  of rank $m$  over the  	ring of
smooth functions on $M$
\begin{equation}\label{eq:Hmod}
\mathcal{H}_{\big( f\boldsymbol{K_i} +
                             g\boldsymbol{K}_j\big)}(X,Y)= 0
 \ , \qquad\forall X, Y \in TM \ ,\forall f,g \in C^\infty(M)\ ;
\end{equation}
  \item
 a ring  w.r.t. the composition operation
\begin{equation} \label{eq:Hring}
\mathcal{H}_{\big(\boldsymbol{K}_i \, \boldsymbol{K}_j\big)}(X,Y)=
 \mathcal{H}_{\big(\boldsymbol{K}_j \, \boldsymbol{K}_i\big)}(X,Y)=0\ , \qquad
\forall \boldsymbol{K}_i,\boldsymbol{K}_j\in  \mathscr{H} , \quad\forall X, Y \in TM\ ,
\end{equation}
\end{itemize}
\end{itemize}

The assumption \eqref{eq:Hmod},  \eqref{eq:Hring}, ensure that the set $\mathscr{H}$ is an associative algebra of Haantjes operators.
In addition, if
\begin{equation}
\boldsymbol{K}_i\,\boldsymbol{K}_j=\boldsymbol{K}_i\,\boldsymbol{K}_j\, \qquad\forall \boldsymbol{K}_i,\boldsymbol{K}_j \in  \mathscr{H}\ ,
\end{equation}
the  algebra $\mathscr{H}$ will be said an Abelian Haantjes algebra.

 \end{definition}
\par

\begin{remark}\label{rem:Hdiag}
The conditions of Definition \ref{def:HM} are apparently very demanding and difficult to solve. However, a class of natural solutions is given, in a local chart
$\{ U,\boldsymbol{x}=(x_1,\ldots,x_n)\}$,  by each operator of the form

\begin{equation}\label{eq:Hdiag}
\boldsymbol{K}=\sum _{i=1}^n l_{i }(\boldsymbol{x})
\frac{\partial}{\partial x_i}\otimes \rd x_i  \ .
 \end{equation}
 The  diagonal operators $\boldsymbol{K}$ have their Haantjes tensor vanishing  and satisfy the differential compatibility condition \eqref{eq:Hmod} by virtue of Proposition \ref{pr:Hd}.  Moreover, they form
a commutative ring, therefore they fulfill Eqs. \eqref{eq:Hring}.
 In fact, such operators generate an algebraic structure that we shall call  a \emph{diagonal} Haantjes algebra.
\end{remark}

A particular but especially relevant class of Haantjes algebras  is given by
the ones generated by  a {\it single} Haantjes operator   $\boldsymbol{L}:TM\mapsto TM$. In fact,  one can construct directly a Haantjes algebra $\mathcal{L}$, of   $rank \leq n=dim(M)$, by choosing as a set of generators  the   first $(n-1)$ powers of $\boldsymbol{L}$ together with $\boldsymbol{L}^0:=\boldsymbol{I}$
\begin{equation*}  \label{algebraic}
\mathcal{L}(\boldsymbol{L}):=Span\{\boldsymbol{I} , \boldsymbol{L},\boldsymbol{L}^2, \boldsymbol{L}^{n-1}\}\ .
\end{equation*}
We shall call these algebras \emph{cyclic} Haantjes algebras. Their rank is equal to the degree of the minimal polynomial of $\boldsymbol{L}$.
\par

A natural question is to establish when  a given Haantjes algebra  can be generated by a single Haantjes  operator, giving rise to a cyclic Haantjes algebra.  This problem has been  investigated  in \cite{TT2017}  starting from the following
\begin{definition} \label{def:CHa}
Let  $(M,  \mathscr{H})$ be a Haantjes algebra of rank $m$. An  operator $\boldsymbol{L}$ will be called a cyclic generator of  $\mathscr{H}$  if
\begin{equation*}
 \mathscr{H}\equiv\mathcal{L}(\boldsymbol{L})
\end{equation*}
  The basis
 \begin{equation} \label{eq:baseCicl}
 \mathcal{B}_{cyc}=\{ \boldsymbol{I} , \boldsymbol{L},\boldsymbol{L}^2, \boldsymbol{L}^{m-1} \}
 \end{equation}
   will be called a cyclic basis of $\mathscr{H}$ and allows us to represent
each Haantjes operator $\boldsymbol{K}\in \mathscr{H} $  as a polynomial field in $\boldsymbol{L}$ of degree at most (m-1), i.e.
\begin{equation}\label{eq:Hg}
\boldsymbol{K}=p_{\boldsymbol{K} }(\boldsymbol{x},\boldsymbol{L})=\sum_{i =0} ^{m-1} a_i(\boldsymbol{x})\,\boldsymbol{L}^i \ ,
\end{equation}
where $a_i(\boldsymbol{x})$ are  smooth functions in $M$.
\end{definition}

\section{Poisson--Haantjes manifolds}
In order to deal with GZ systems, we need to  extend the notion of symplectic-Haantjes manifold ($\omega\mathcal{H}$) already introduced in \cite{TT2016prepr}. Here we propose the new notion of  Poisson--Haantjes ($P\mathcal{H}$) manifold. 
\par
As usual, the transposed operator  $\boldsymbol{K}^{T}: T^*M\rightarrow T^*M$ is  defined as the transposed linear map of $\boldsymbol{K}:TM\rightarrow TM$, with respect to the natural pairing between a vector space and its dual space.

\begin{definition}\label{def:PHman}
A Poisson--Haantjes manifold is a triple $( M,P,\mathscr{H})$ that satisfies the following conditions
\begin{itemize}
\item[i)]
M is a differentiable manifold;
\item[ii)]
$P:TM^* \rightarrow TM$ is a Poisson bivector in $M$;
\item[iii)]
$\mathscr{H}$ is an Abelian Haantjes algebra;
\item[iv)]
$(P,\mathscr{H})$ are algebraically compatible in the sense that 
$\boldsymbol{K} P=P\boldsymbol{K}^T , \ \forall \boldsymbol{K} \in \mathscr{H}$.
\end{itemize}
\end{definition}
As a consequence of the above conditions, we get the following simple proposition that turns out to be crucial for many results of the present theory.

\begin{proposition}\label{pr:ss}
In a given a $P \mathcal{H}$ manifold,  any composed operator,
 $\boldsymbol{K}_i\, P$,
 $\boldsymbol{K}_i \,P \,\boldsymbol{K}_j^T$,
   $\big( \boldsymbol{K}_\alpha-f(\boldsymbol{x})\boldsymbol{I}\big)^r P$, $r \in \mathbb{N}$,   is skew symmetric.
\end{proposition}
\begin{remark}
The class of $\omega\mathcal{H}$ manifolds  coincides  with the  class of $P\mathcal{H}$ manifolds of even dimension, with an invertible  Poisson bivector. In fact, as in this case $\Omega=P^{-1}$ is a symplectic operator, the compatibility condition $\boldsymbol{K} P=P\boldsymbol{K}^T $ is equivalent to the compatibility condition $\Omega\boldsymbol{K} =\boldsymbol{K}^T \Omega$, required in $\omega\mathcal{H}$ manifolds.
\end{remark}
We show a paradigmatic example of $ P \mathcal{H}$ manifold with a cyclic Haantjes algebra, that later will be used to describe a Haantjes algebra of the Lagrange top.
\begin{example} \label{ex:PN}
Let $(M,P,\boldsymbol{N})$ be a Poisson--Nijenhuis (PN) manifold, that is,  a manifold endowed with a Poisson bivector $P$ and a Nijenhuis operator 
$\boldsymbol{N}$ that satisfy the following compatibility conditions
\begin{eqnarray}
\label{eq:PNa}
\boldsymbol{N}P-P\boldsymbol{N}^T&=& 0 \\
\label{eq:PNd}
R(P,\boldsymbol{N})(\alpha,Y)&=&0 \qquad\qquad \forall \alpha \in T^*M \ , \forall Y \in TM \ ,
\end{eqnarray}
where $R(P,\boldsymbol{N})$ is the $(2+1)$ tensor field defined in \cite{MM} by 
\begin{equation}
R(P,\boldsymbol{N})(\alpha,Y)=\mathcal{L}_{P\alpha}\big(\boldsymbol{N} \big)Y-P\big ( \mathcal{L}_Y\big(\boldsymbol{N}^T\alpha\big)-\mathcal{L}_{\boldsymbol{N}Y}\alpha \big)\ ,
\end{equation}
($\mathcal{L}_Y$ denotes the Lie derivative with respect the vector field $Y$).
Let us suppose  that the Nijenhuis operator $\boldsymbol{N}$ has its minimal polynomial  of degree 
$\mathrm{m}$. 
Then, the $PN$  manifold $M$ has a \emph{standard} $ P \mathcal{H}$ structure, given by
$$
(M, P,  \boldsymbol{K}_1= \boldsymbol{I},  \boldsymbol{K}_2= \boldsymbol{N},\ldots,   , \boldsymbol{K}_{m}= \boldsymbol{N^{m-1}}) \ ,
$$
with a Haantjes algebra of rank  $\mathrm{m}\leq dim(M)$.
In fact, each Nijenhuis operator $ \boldsymbol{N}$ is also a Haantjes operator, therefore generates the  cyclic Haantjes algebra $\mathcal{L}(\boldsymbol{N})$. Moreover, the algebraic compatibility condition \eqref{eq:PNa} assures that for all Haantjes operators 
$$
\boldsymbol{K}=p_{\boldsymbol{K} }(\boldsymbol{x},\boldsymbol{N})=\sum_{i =0} ^{m-1} a_i(\boldsymbol{x})\,\boldsymbol{N}^i \ ,
$$
 the condition iv) of Def. 9 is fulfilled.
 \par
 In addition, the differential  condition \eqref{eq:PNd} implies that 
 \begin{equation}
R( P,\boldsymbol{K})(\alpha,Y)- \Bigg(\sum_{i=0}^{m-1}\bigg( \big(X_i\wedge \boldsymbol{N}^iY\big)- Y(a_i)\boldsymbol{N}^{i} P\bigg)\Bigg) \alpha=0   
\end{equation}
 $\forall j\in \mathbb{N}$, $\forall \alpha \in T^*M$, $\forall Y \in TM$, where $X_i:=P\rd{a_i}$ are the Hamiltonian vector fields with Hamiltonian functions $a_i$.
\end{example}
 We also generalize the concept of Magri--Haantjes chain, introduced in \cite{TT2016prepr} under the name of Lenard--Haantjes chains.
\begin{definition}
 Let $( M,P,\mathscr{H})$ be a Poisson--Haantjes manifold. A function  $H\in C^\infty(M)$ is said to generate  a Magri--Haantjes chain of 1-forms if
$$
\rd(\boldsymbol{K}^T_i \,\rd H)=0,  \quad\qquad i=1,\ldots ,m ,
$$
for some basis \{$\boldsymbol{K}_1,\ldots, \boldsymbol{K}_m\}$ of   $\mathscr{H}$. The (locally) exact 1-forms $\rd H_i$ such that
$$
\rd H_i=\boldsymbol{K}^T_i \,\rd H   \quad\qquad i=1,\ldots ,m ,
$$
are called the elements of the Magri--Haantjes chain, of length $m$, generated by $H$.
\end{definition}
The relevance of Magri--Haantjes chains is due to the following
\begin{lemma}
 Let $( M,P,\mathscr{H})$ be a Poisson--Haantjes manifold.  The functions $H_i$ whose differentials belong to  all Magri-Haantjes chains  generated by a single function $H$ are in involution w.r.t. the Poisson bracket defined by $P$. In fact,
\begin{equation}
\{H_i, H_j\}=<dH_i, P\, dH_j>=<\boldsymbol{K}_i^T dH,P\boldsymbol{K}^T_{j} dH>=
<dH,\boldsymbol{K}_{i} P \boldsymbol{K}_{j}^T dH>\stackrel{Prop.\, \ref{pr:ss}}{=}0
\end{equation}
\end{lemma}
\begin{definition}
  Let $( M,P,\mathscr{H})$ be a Poisson--Haantjes manifold. A  vector field $Y$ is said to generate a Magri--Haantjes chain of vector fields if the vector fields defined by 
  $$
Y_i:=\boldsymbol{K}_i \, Y  ,\qquad\qquad i=1,\ldots,m
$$ 
for some basis \{$\boldsymbol{K}_1,\ldots, \boldsymbol{K}_m\}$, commute among each others.
\end{definition}
\begin{remark} \label{rem:McXH}
Let us note that, thanks to the compatibility condition $iv)$ in Definition \ref{def:PHman},  to every Magri-Haantjes chain of 1-forms generated by a function $H$ corresponds a Magri-Haantjes chain of Hamiltonian vector fields $X_{H_i}=P\rd H_i$ generated by $X_H=P\rd H$.  Moreover, the Hamiltonian vector fields belonging to different chains generated by the same Hamiltonian vector field $X_H$ commute among each others.
\end{remark}
In \cite{TT2016prepr}, it has been shown that, given a $\omega\mathcal{H}$  manifold and a function $H$, the existence of a Magri--Haantjes chain generated by $H$ is equivalent to the Frobenius integrability of the co-distribution
\begin{equation}
\boldsymbol{K}^T_i \rd H    \qquad  \qquad  \qquad  i=1,\ldots m.
\end{equation}
In this paper, we limit ourselves to exhibit the example of the Magri-Haantjes chain of the Lagrange top, leaving the finding of  the  conditions assuring the existence of  Magri-Haantjes chains in  $P\mathcal{H}$ manifolds to future work.

\section{The Lagrange top}
The classical Lagrange top is a heavy symmetric top, that is, a symmetric rigid body with a fixed point $O$, immersed  in the uniform gravity field $\vec{\gamma}$. It  admits different geometric formulations in the framework of the bi--Hamiltonian theory
\cite{MTlt,MTltC, TSIlt}.
\begin{center}
\includegraphics[scale=0.8]{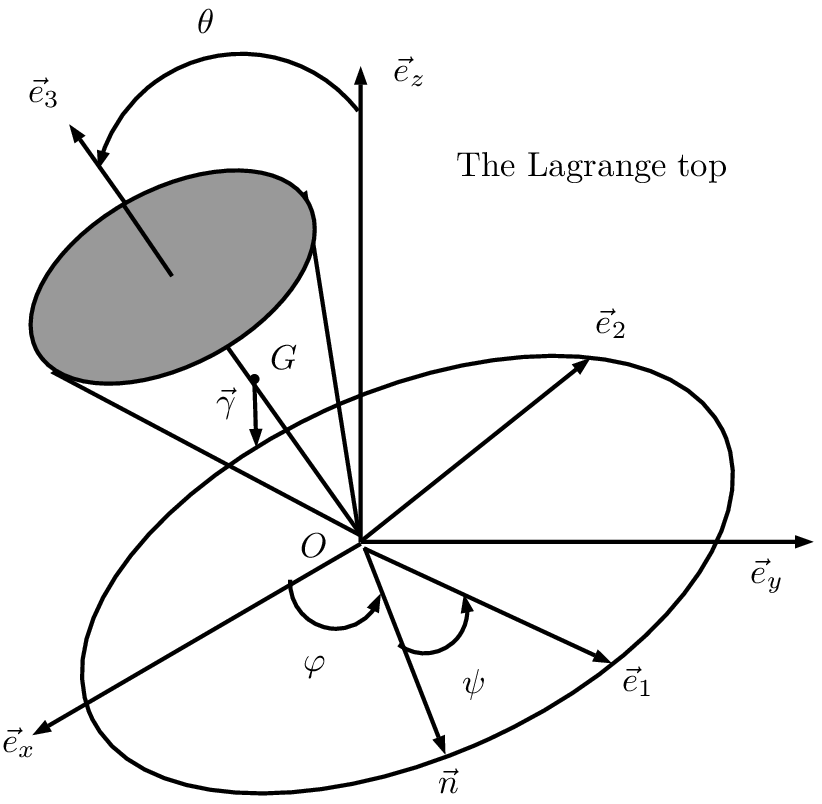}
\end{center}
\subsection{Euler angles}
In the phase space $M=T^*(SO(3))$ one can choose as local coordinates the classical  Euler angles and conjugate momenta $( \varphi,\theta,\psi,p_\varphi,p_\theta,p_\psi )$.  In such a chart, the Hamiltonian function of the Lagrange top takes the form 
\begin{equation} \label{eq:HL}
H=\frac{1}{2\mathcal{A}} \bigg(p_\theta^2+\frac{(p_\varphi-p_\psi\cos\theta)^2}{\sin^2\theta}+\frac{1}{c}p_\psi^2\bigg)+\mu ga \cos\theta
\end{equation}
and the Hamilton equations of the motion are
\begin{align*}
   \dot{\varphi}  =&\frac{1}{\mathcal{A}} \frac{p_\varphi-p_\psi\cos\theta }{\sin^2\theta}   \\
    \dot{\theta}  =&\frac{1}{\mathcal{A}}\, p_\theta \\
   \dot{\psi}  =&\frac{1}{c\mathcal{A}} \,p_\theta \\
   \dot{p}_\varphi=&0\\
   \dot{p}_{\theta}=&-\frac{\partial}{\partial \theta}\,\frac{(p_\varphi-p_\psi\cos\theta)^2}{\sin^2\theta}+\mu ga\sin\theta\\
  \dot{p}_\psi  =&0 \ ,
\end{align*}
where $\mathcal{A}$ and $c \,\mathcal{A}$ are, respectively, the inertia momenta w.r.t. every axis in the equatorial plane and the symmetry axis,   $\mu$ is the mass of the top, $g$ the gravity acceleration, $a$ the coordinate of the mass center $G$ along the symmetry axis.
It is evident from the Hamiltonian function and the equations of the motion that the Lagrange top admits the three integrals of motion
\begin{equation} \label{eq:HiL|}
H_1=H\ ,\qquad H_2=p_\varphi \ ,\qquad H_3=p_\psi \ ,
\end{equation}
that are, the energy and the components of the angular momentum along the vertical  and the symmetry axis, respectively. Moreover, it is well known that the coordinates $( \varphi,\theta,\psi,p_\varphi,p_\theta,p_\psi )$ are separation variables for the Hamilton-Jacobi equation associated to $H$. From our point of view,  it is worth of interest to show that even this very classical system, as well as every Hamiltonian separable system, can be described in the framework of $\omega \mathcal{H}$  Haantjes manifolds,   according to Theorem 57 of \cite{TT2016prepr}.  A Haantjes algebra for the Lagrange top can be easily computed, whose basis is  

\begin{eqnarray*} \label{eq:LSoV}
\boldsymbol{K}_1&=&\boldsymbol{I}
 \\
\boldsymbol{K}_2&=& \frac{\mathcal{A} \sin ^2\theta}{p_\varphi-p_\psi\,\cos\theta}\bigg (\frac{\partial}{\partial \varphi}\otimes \rd \varphi+\frac{\partial}{\partial p_\varphi}\otimes \rd p_\varphi \bigg ) \\
\boldsymbol{K}_3&=&- \frac{\mathcal{A} \sin ^2\theta}{\cos\theta(p_\varphi-p_\psi\,\cos\theta)}\bigg (\frac{\partial}{\partial \theta}\otimes \rd \theta+\frac{\partial}{\partial p_\theta}\otimes \rd p_\theta \bigg ) \ .
\end{eqnarray*}
In fact, the action of the (transpose of the) Haantjes operators $\boldsymbol{K}_1, \boldsymbol{K}_2,\boldsymbol{K}_3$ on the gradient of the Hamiltonian function \eqref{eq:HL} produces the Magri--Haantjes chain of the (gradients of the) three integrals of motion
$$
\boldsymbol{K}_1^T \rd H= \rd H_1 \ ,\qquad\boldsymbol{K}_2^T \rd H=\rd H_2\ ,\qquad \boldsymbol{K}_3^T\rd H=\rd H_3 \ .
$$
The fact that the Haantjes operators take a diagonal form in the Euler chart $( \varphi,\theta,\psi,p_\varphi,p_\theta,p_\psi )$, means that the Haantjes algebra is diagonal and  the Euler coordinates are Darboux-Haantjes coordinates for it.

 More interesting and hard to solve is the problem to construct a Haantjes algebra directly in the physical coordinates, and, afterwards determining a set of Darboux--Haantjes coordinates that are separation variables according to Theorem 59 of \cite{TT2016prepr}. In the next Section, we will show how to proceed in this example, starting with  the tri-Hamiltonian formulation of the Lagrange top and reducing it to the symplectic leaves of one of its Poisson bivectors.
\subsection{Euler--Poisson coordinates}
An alternative formulation of the Lagrange top is based on the Euler-Poisson equation that are, roughly speaking,  the equation of the motion projected onto the comoving reference frame $(\vec{e}_1,\vec{e}_2,\vec{e}_3)$.
Let us consider the phase space $M:=\{m | m=(\vec{\omega},\vec{\gamma})\}$, where 
$\vec {\omega}= \omega_1\vec{e}_1+ \omega_2 \vec{e}_2+ \omega_3\vec{e}_3$ is the angular velocity of the top, 
$\vec{\gamma} =\gamma_1\vec{e}_1+\vec{e}_2 \gamma_2+\gamma_3 \vec{e}_3$    the vertical unit vector (the Poisson vector),
$G-O=a\vec{e}_3$  the vector of the center of mass and  $J=\mbox{diag}(\mathcal{A}, \mathcal{A}, c\mathcal{A})$ the principal inertia matrix. 
In such a formulation, the equations of the motion $ \dot{m}=X_L(m)$ are the Euler equations coupled with the Poisson equations, so that the Lagrange vector field is given by
\begin{equation} \label{eq:XL}
X_L(m)= \left[
\begin{array}{c}
(1-c) \omega_2 \omega_3-\gamma_2 \\
-(1-c) \omega_3 \omega_1+\gamma_1 \\
0\\
\gamma_2 \omega_3- \gamma_3 \omega_2\\
\gamma_3 \omega_1- \gamma_1 \omega_3\\
\gamma_1 \omega_2- \gamma_2 \omega_1\\
\end{array}
\right] \ ,
\end{equation}
where the normalization $\mu  g a/\mathcal{A}=1$ has been chosen.
\subsubsection{The tri-Hamiltonian formulation of the Lagrange top}
The vector field \eqref{eq:XL} admits a tri-Hamiltonian formulation (see \cite{MTlt} and reference therein) w.r.t. the three non invertible Poisson bivectors
$$
P_0= \left[
\begin{array}{cc}
0  & B \\
B & C \\
\end{array}
\right]
\ , \ 
P_1= \left[
\begin{array}{cc}
-B  & 0 \\
0 & \Gamma \\
\end{array}
\right]
\ ,\ 
P_2= \left[
\begin{array}{cc}
T  & R \\
-R^T & 0 \\
\end{array}
\right] \ ,
$$ 
where
$$
B= \left[
\begin{array}{ccc}
       0  & -1 & 0 \\
1 & 0 & 0 \\
0 & 0& 0 \\
\end{array}
\right]
\ , \ 
C= \left[
\begin{array}{ccc}
       0  & c ~\omega_3 & -\omega_2 \\
- c ~\omega_3 & 0 & \omega_1 \\
\omega_2 & -\omega_1 & 0 \\
\end{array}
\right]
$$
$$
\Gamma = \left[
\begin{array}{ccc}
       0  & \gamma_3 & -\gamma_2 \\
-\gamma_3 & 0 & \gamma_1 \\
\gamma_2 & -\gamma_1 & 0 \\
\end{array}
\right]
\ ,\ 
R = \left[
\begin{array}{ccc}
       0  & -\gamma_3 & \gamma_2 \\
\gamma_3 & 0 & -\gamma_1 \\
-\gamma_2/ c & \gamma_1 /c & 0 \\
\end{array}
\right]
$$
$$
T= \left[
\begin{array}{ccc}
       0  & - c~ \omega_3 & \omega_2/ c \\
       c~ \omega_3 & 0 & - \omega_1 / c \\
- \omega_2 / c & \omega_1/ c & 0 \\
\end{array}
\right]\ ,
$$
and the Hamiltonian functions $(h_0,h_1,h_2)$
$$
\xymatrix{
dh_0  
\ar[dr]_{P_0}
&dh_1  
\ar[d]|{P_1}
&dh_2 
\ar[dl]^{P_2} \\
&X_L& 
}
$$
defined by
$$
h_0=\frac{1}{2} F_4+(c-1) F_1 F_3, \quad
h_1= -F_3-(c-1) F_1 F_2, 
\quad
h_2=F_2 \ .
$$ 
The functions $F_1$, $F_2$, $F_3$, $F_4$ are the integrals of motion given by
$$
F_1=\omega_3, \quad
F_2=\frac{1}{2}(\omega_{1}^2+\omega_{2}^2+c~ \omega_{3}^2)-\gamma_3, 
$$
$$
F_3= \omega_1\gamma_1+\omega_2\gamma_2+ c~ \omega_3\gamma_3~,
\quad F_4= \gamma_{1}^2+\gamma_{2}^2+ \gamma_{3}^2~.
$$
The three Poisson bivectors $(P_0,P_1,P_2)$ generate three Poisson pencils
$$
P_0-\lambda P_1 \ , \qquad P_1-\lambda P_2\ ,\qquad P_0-\lambda P_2 \ ,
$$
that possess two polynomial Casimir functions each \cite{MTltC}. Below, we concentrate on the Casimir function of the Poisson pencil $P_0-\lambda P_1$. 
\subsubsection{A Gelfand--Zakharevich system of co--rank 2}
  The Poisson pencil $P_0-\lambda P_1$ possess  two polynomial Casimir functions $H(\lambda)^{(1)}=H_0^{(1)}=F_1$ and $H(\lambda)^{(2)}=H_0^{(2)} \lambda^2+H_1^{(2)}\lambda+H_2=\frac{F_4}{2} \lambda^2-F_3\lambda+F_2$,
of length $1$ and $3$ respectively. They    can be represented graphically in the following way
$$
\xymatrix{
 &dF_1
\ar[dl]_{P_0} \ar[dr]^{P_1}&\\
0&&0
}
$$
\begin{equation}\label{eq:2BHchain}
\xymatrix{
&\rd F_4/2 \ar[dl]_{P_1} \ar[dr]^{P_0}
& &\rd(-F_3)\ar[dl]_{P_1} \ar[dr]^{P_0}
& &\rd(F_2)\ar[dl]_{P_1} \ar[dr]^{P_0}& \\
0& &X_1& &X_2& &0&
} \ .
\end{equation}
The vector fields $X_1,X_2$ are bi--Hamiltonian vector fields as 
\begin{equation} \label{eq:XiBH}
X_i=P_0 \rd H_{i-1}^{(2)}=P_1  \rd H_{i }^{(2)}\ ,\qquad i=1, 2 \ .
\end{equation}
Moreover, the vector field \eqref{eq:XL} of the Lagrange top can be formulated as
\begin{equation}\label{eq:XLGZ2}
X_L=X_1-(c-1)F_1\, X_2 \ ,
\end{equation}
therefore defining  a system of Gelfand--Zakharevich type of co--rank $2$.
\subsubsection{The reduction of the Poisson pencil}
Without loss of generality, we fix a Poisson bivector inside the Poisson pencil, say $P_1$, and, in order to getting rid of its Casimir functions, we perform a reduction to its symplectic leaves
\begin{equation} \label{eq:S1}
S_1:=\{ F_1=C_1, F_4=C_4\}\ . 
\end{equation}
To this aim, it is convenient to introduce complex coordinates in $M$ adapted to such a reduction \cite{MTlt}
$$
x_1=-c \omega_3+i\omega_2 \ , \quad x_2=\gamma_3-i\gamma_2 \ ,\quad y_1= \omega_1\ ,\quad y_2 =-\gamma_1\ , \quad F_1\ , \quad F_4 \ ,
$$
 in which the Poisson bivectors take the following simple form
 $$
P_0=
\left[
\begin{array}{c|c}
 \check{P}_0&
\begin{array}{c|c}
0\quad & 2 X_1^1 \\
0 \quad&2  X_1^2\\
0\quad&2 X_1^3\\
0\quad&2 X_1^4 \\
\end{array}\\
 \hline
\begin{array}{cccc}
0 & 0&0&0 \\
-2X_1^1 & -2X_1^2&-2X_1^3&-2X_1^4
\end{array}&
 \begin{array}{c|c}
\!\!\! \!\!0 &2 X_1^5 \\
 \!\!\!\!-2X_1^5&0
 \end{array}
 \end{array}
 \right ] \ ,
\quad
\check{P}_0=-i
\left[
\begin{array}{cc|cc}
0&0&0&1\\
0&0&1&-x_1\\
\hline
0&-1&0&0\\
-1&x_1&0&0
\end{array}
\right ] 
$$

$$
P_1=
\left[
\begin{array}{c|c}
\check{P}_1&
\begin{array}{c|c}0 & 0 \\
0&0\\
0&0\\
0 & 0
\end{array}\\
 \hline
\begin{array}{cccc}
0 & 0 &0&0\\
0 & 0&0&0
\end{array}&
 \begin{array}{c|c}
 0 & 0 \\
0 & 0
 \end{array}
 \end{array}
 \right ] \ ,
\quad
\check{P}_1=-i
\left[
\begin{array}{cc|cc}
0&0&1&0\\
0&0&0&x_2\\
\hline
-1&0&0&0\\
0&-x_2&0&0
\end{array}
\right ] \ .
$$
It is evident, that the Poisson bivector $P_0$ cannot be restricted to $S_1$, unlike $P_1$. So, we perform a reduction procedure introduced in \cite{FPrmp}. It is a highly non trivial generalization of the Marsden--Ratiu method \cite{MR} and consists in an  ingenious deformation of the Poisson bivector  $P_0$, by means of a suitable family of vector fields transversal to the symplectic leaves of $P_1$. Such a deformation assures that the  deformed Poisson pencil shares the same \textit{axis}  with the old one  and can be restricted to $S_1$.
\subsubsection{Deformation}
We choose the two vector fields 
 $$
 Z_1=\frac{\partial}{\partial F_1}\ ,\qquad Z_2=2 \frac{\partial}{\partial F_4}
 $$
normalized as 
$$
Z_i (H_0^{(j)})=\delta_i^j \ , \qquad i,j=1,2 \ .
$$
As they fulfill the equations 
$$ 
\mathcal{L}_{Z_i}(P_1)=0 \qquad \mathcal{L}_{Z_i}(P_0)=[Z_i,X_1]\wedge Z_2 \qquad i=1,2 \ ,
$$
they  can deform  the  Poisson bivector $P_0$ into the new Poisson bivector 
\begin{equation}\label{eq:Qdef}
Q:=P_0-X_1\wedge Z_2 \ , \quad Q=
\left[
\begin{array}{c|c}
 \check{P}_0&
\begin{array}{c|c}0 & 0 \\
0&0\\
0&0\\
0 & 0
\end{array}\\
 \hline
\begin{array}{cccc}
0 & 0 &0&0\\
0 & 0&0&0
\end{array}&
 \begin{array}{c|c}
 0 & 0 \\
 0 & 0
 \end{array}
 \end{array}
 \right ] \ ,
\end{equation}
that  can be restricted to $S_1$ and its restriction is $\check{Q}=\check{P_0}$.
\par
\subsubsection{Haantjes algebra } \label{sec:HaLT} 
There are different manners of endowing the manifold $M$ with a Haantjes algebra compatible with $P_1$. In this section, we limit ourselves to show one of them, leaving a more general  discussion to future work.  In order to construct a Haantjes algebra  for the bi--Hamiltonian chain \eqref{eq:2BHchain}, which will be preserved after  the restriction to $S_1$, we look for a Nijenhuis operator $\boldsymbol{N}$ that plays the role of a cyclic generator for a Haantjes algebra of low rank, that is  $2\leq  m\leq 3$. Precisely, we require that $\boldsymbol{N}$:
\begin{itemize}
\item[i)]
factorizes the deformed Poisson bivector 
\begin{equation} \label{eq:QNP}
Q=\boldsymbol{N} P_1 \ ;
\end{equation}
\item[ii)]
has its restriction $\boldsymbol{\check{N}} $ to $S_1$ equal to
\begin{equation} \label{eq:Ncheck}
\boldsymbol{\check{N}} =\check{Q}\check{P}_1^{-1}=\check{P}_0\check{P}_1^{-1} \ ;
\end{equation}
\item[iii)]
 is a cyclic generator of a Haantjes algebra that provides the following Magri--Haantjes  chain of vector fields generated by $X_1$ 
 \begin{equation} \label{eq:MHcampi}
 \boldsymbol{K}_i X_1=X_i \ , \qquad i=1,2,3 \ ,
 \end{equation}
where $X_3=0$.
 \end{itemize}
 In other words, denoted with $\Phi$ the immersion of $S_1$ in $M$, and with $\Phi_*$ and $\Phi^*$ its pullback and pushforward respectively, we search for an  operator $\boldsymbol{N}: TM\rightarrow TM$ that solves the system
\begin{eqnarray}
\label{eq:QNP1}
\boldsymbol{N}P_1&=&P_0-X_1\wedge Z_2\ \\
\label{eq:Nrest}
\Phi_*\boldsymbol{N}\Phi^*&=&\check{P}_0\check{P}_1^{-1} \\
\label{eq:MHX1}
(d_1 \boldsymbol{I}+e_1 \boldsymbol{N}+f_1 \boldsymbol{N}^2 ) X_1&=&X_1 \\
\label{eq:MHX2}
(d_2 \boldsymbol{I}+e_2 \boldsymbol{N}+f_2 \boldsymbol{N}^2 ) X_1&=&X_2 \\
\label{eq:MHX3}
(d_3 \boldsymbol{I}+e_3 \boldsymbol{N}+f_3 \boldsymbol{N}^2 ) X_1&=&X_3 \\
\label{eq:NtorsionLT}
\tau(\boldsymbol{N})&=&0 \ ,
\end{eqnarray}
together with the unknown functions $(d_i,e_i,f_i) ,\  i=1,2,3$. 
Such a system can be decoupled  in order to  firstly  determine the unknown functions $(d_i,e_i,f_i)$.
In fact, applying both terms of Eq. \eqref{eq:QNP1} to the gradients of the bi--Hamiltonian chain \eqref{eq:2BHchain}, and    taking into account  the fact that  the Hamiltonian functions $H_j$ are integrals of motion for $X_1$, one finds that the bi-Hamiltonian vector fields  \eqref{eq:XiBH} must  fulfill the system
\begin{equation} \label{eq:sysK}
\boldsymbol{N}X_j=X_{j+1}+\cancel{X_1(H_j)}\,Z_2-Z_2(H_j^{(2)})\,X_1\ , \qquad X_0:=0 \ ,\qquad j=0,1,2 \ .
\end{equation}
   By solving recursively such a  system  w.r.t. $X_{j+1}$, only in terms of the  vector field $X_1$, a unique solution is found  for the unknown functions $(d_i,e_i,f_i)$ in Eqs \eqref{eq:MHX1}, \eqref{eq:MHX2}, \eqref{eq:MHX3}. This solution can be written down  in a compact form as   
\begin{equation} \label{eq:BenRelM70}
\begin{bmatrix}
  \boldsymbol{K}_1   \\
     \boldsymbol{K}_2 \\
  \boldsymbol{ K}_3  \\
\end{bmatrix}
=
\begin{bmatrix}
1&0&0\\
   Z_2(-F_3) &1&0 \\
Z_2(F_2)& Z_2(-F_3) &1\\
\end{bmatrix}
\begin{bmatrix}
     \boldsymbol{I}   \\
     \boldsymbol{ N }  \\
    \boldsymbol{  N}^2   \\
 \end{bmatrix}
 \ .
\end{equation}
Summarizing, the  operators \eqref{eq:BenRelM70} are Haantjes operators that  provide the Magri--Haantjes  chain \eqref{eq:MHcampi} for any solution $\boldsymbol{N}$ of Eqs. \eqref{eq:QNP1}, \eqref{eq:Nrest} and \eqref{eq:NtorsionLT}.  A simple solution of  Eqs. \eqref{eq:QNP1} and \eqref{eq:Nrest}, which leaves invariant both $TS_1$ and $Span\{Z_1,Z_2\}$, is 
\begin{equation}\label{eq:Ncand}
\boldsymbol{N}=
\left[
\begin{array}{c|c}
 \boldsymbol{\check{N}}&
\begin{array}{c|c}0 & 0 \\
0&0\\
0&0\\
0 & 0
\end{array}\\
 \hline
\begin{array}{cccc}
0 & 0 &0&0\\
0 & 0&0&0
\end{array}&
 \begin{array}{c|c}
\boldsymbol{N}_5^5 & \boldsymbol{N}_6^5\\
 \boldsymbol{N}^6_5 & \boldsymbol{N}^6_6
 \end{array}
 \end{array}
 \right ], \ 
\boldsymbol{\check{N}}=\check{P}_0\check{P}_1^{-1}=
\left[
\begin{array}{cc|cc}
0&\frac{1}{x_2}&0&0\\
1&-\frac{x_1}{x_2}&0&0\\
\hline
0&0&0&\frac{1}{x_2}\\
0&0&1&-\frac{x_1}{x_2}
\end{array}
\right ]  \ ,
\end{equation}
were $(\boldsymbol{N}_5^5,\boldsymbol{N}_6^5,\boldsymbol{N}_5^6,\boldsymbol{N}_6^6)$ are arbitrary functions to be determined requiring that the Nijenhuis torsion of $ \boldsymbol{N}$ identically vanishes.  However, we can postpone the solution of the PDE  \eqref{eq:NtorsionLT}, adding to the system \eqref{eq:QNP1}--\eqref{eq:NtorsionLT} a further \emph{algebraic} request,  namely that $\boldsymbol{N}$  satisfies the following Magri--Haantjes  chain of $1$-forms,
\begin{equation}\label{eq:MHgrad}
(Z_2(-F_3)  \boldsymbol{I}+\boldsymbol{N}\big)^T \rd (-F_3)=\rd F_2 \ ,
 \big(Z_2(F_2) +Z_2(-F_3)\boldsymbol{N}+\boldsymbol{N}^2\big)^T\rd (-F_3)=0 \ .
\end{equation} 
 Eqs. \eqref{eq:MHgrad} are independent of Eqs. \eqref{eq:MHX1}--\eqref{eq:MHX3} as, due to the kernel of the Poisson operator, a Magri--Haantjes chain of Hamiltonian vector fields   do not imply that the Magri--Haantjes chain of the gradients of their Hamiltonian functions holds true.  Thus, substituting \eqref{eq:Ncand} into Eqs. \eqref{eq:MHgrad}, we find the 
unique solution 
$$
\boldsymbol{N}_5^5=\frac{(c-1)F_1+x_1} {\Delta} \ ,\quad \boldsymbol{N}_6^5=\frac{1} {2\,c\,x_2\,\Delta}
$$
$$
\boldsymbol{N}_5^6=-\frac{2cx_2\big((c-1)(F_1^2+F_1x_1)-x_2\big)} {\Delta}
$$
$$
\boldsymbol{N}_6^6=-\frac{x_1^3+(c-1)F_1 x_1^2+2 x_1 x_2+(c-1)F_1 x_2} {x_2\, \Delta} \ , 
$$
where $\Delta=x_1^2+(c-1)F_1 x_1+x_2$. Now, we can  verify that it satisfies also  Eq. \eqref{eq:NtorsionLT}, therefore we have got a solution of the problem \eqref{eq:QNP}--\eqref{eq:MHcampi}, which solves also Eqs. \eqref{eq:MHgrad}.
  It has its minimal polynomial  
of degree $2$
\begin{equation}\label{eq:mpN}
m_{\boldsymbol{N}}(\lambda,\boldsymbol{x})=\lambda^2 +\frac{x_1}{x_2}\lambda-\frac{1}{x_2} =\lambda^2 +Z_2(-F_3)\lambda+Z_2(F_2)\ ,
\end{equation}
therefore it is a cyclic  generator of  a Haantjes algebra of rank $2$.  Moreover, the coefficients of \eqref{eq:mpN} coincides with the elements of the last row in the matrix of \eqref{eq:BenRelM70}, consequently  $\boldsymbol{K}_3=0$.

\subsubsection{Restriction to the symplectic leaves of $P_1$}
From the previous steps it follows
\begin{proposition}\label{pr:S6}
The deformed Poisson pencil $Q-\lambda P_1$ restricts to the symplectic leaves of $P_1$. 
Moreover,  $ \boldsymbol{N}$, $ \boldsymbol{K}_1$, $ \boldsymbol{K}_2$,  the Hamiltonian functions and the Hamiltonian vector fields $(X_1,X_2)$ restrict as well. They endow the symplectic leaves of $P_1$ with a $\omega \mathcal{H}$ structure and   the Magri--Haantjes chain of exact forms, given by
\begin{equation} \label{eq:MHcF3F2}
\qquad \boldsymbol{\check{K}}_1^T\, \rd (-F_3)_{|{S_1}}=\rd (-F_3)_{|{S_1}}\ ,\qquad \boldsymbol{\check{K}}_2^T\, \rd (-F_3)_{|{S_1}}=\rd F_{2_{|{S_1}}} \ ,
\end{equation}
 in virtue of   \eqref{eq:MHgrad} .
\end{proposition}
In particular, the relations (\ref{eq:BenRelM70}) restricts to the Benenti relations 
$$
\boldsymbol{\check{K}}_1=\boldsymbol{\check{I}}\ , \quad \boldsymbol{\check{K}}_2={\frac{x_1}{x_2} }\boldsymbol{\check{I}}+\boldsymbol{\check{N}}  \ ,\quad
$$
as can be immediately seen from the analysis of the minimal polynomial of $\boldsymbol{\check{N}}$ 
that is still equal to \eqref{eq:mpN}. Such relations are satisfied by the so-called $\boldsymbol{L}$-systems \cite{Ben92,Ben93,BenAAM}, proved  to be projections of Quasi--Bi--Hamiltonian  \cite{MT,MTPLA,MTRomp} systems constructed in Riemannian manifolds \cite{TT2016}. 

\subsubsection{Separation of variables}
Let us construct in $S_1$ a set of Darboux-Haantjes (DH) coordinates for the $\omega \mathcal{H}$ manifold $(S_1,P_1,\boldsymbol{\check{K}}_1,\boldsymbol{\check{K}}_2)$.  As the $\omega\mathcal{H}$ structure, in this case, is equivalent to a $\omega N$ structure, as DH coordinates we can take just a set of Darboux-Nijenhuis coordinates \cite{FP}. To this aim, we proceed according to the Remark 70, Sect. 8 of \cite{TT2016prepr}. Firstly, we choose as  first two  coordinates $(\lambda_1,\lambda_2)$  just the two (double) eigenvalues  of the Haantjes operator  
$\boldsymbol{\check{K}}_2$
$$
\lambda_1= \frac{x_1-\sqrt{x_1^2+4x_2}}{2 x_2} \ ,\qquad
\lambda_2= \frac{x_1+\sqrt{x_1^2+4x_2}}{2 x_2} \ ,
$$
as their gradients are eigenforms of $\boldsymbol{\check{K}}_2^T$
$$
\boldsymbol{\check{K}}_2^T\rd \lambda_1=\lambda_2\, \rd \lambda_1 \ ,\qquad
\boldsymbol{\check{K}}_2^T\rd \lambda_1=\lambda_2\, \rd \lambda_1 \ .
$$
Let us note that $(\lambda_1,\lambda_2)$ are also the only eigenvalues of $\boldsymbol{K}_2=Z_2(-F_3)\boldsymbol{I}+\boldsymbol{N}$,  which, therefore, can be used to find half of the separation variables. 
\par
Further, we complete them with a pair of conjugate momenta 
$$
\mu_1= \frac{1}{\lambda_1} (\lambda_2 y_1+y_2)
\qquad
\mu_2=\frac{1}{\lambda_2} (\lambda_1 y_1+y_2)
$$
whose gradients  are also eigenforms of $\boldsymbol{\check{K}}_2$ 
$$
\boldsymbol{\check{K}}_2^T\rd \mu_1=\lambda_2\, \rd \mu_1 \ ,\qquad
\boldsymbol{\check{K}}_2^T\rd\mu_1=\lambda_2\, \rd \mu_2 \ .
$$
The local chart  $(\lambda_1,\lambda_2,\mu_1,\mu_2)$ so constructed is a Darboux chart for the Poisson operator $\check{P}_1$ and a Haantjes chart for the Haantjes operator $\boldsymbol{\check{K}}_2$. In fact, in such chart, they take the form 
$$
\check{P}_1=i
\left[\begin{array}{cc|cc}
0 & 0 & 1 & 0 \\
0 & 0 & 0 & 1 \\\hline 
-1 & 0 & 0 & 0 
\\ 0 & -1 & 0 & 0
\end{array}\right] 
\ ,\quad
\boldsymbol{\check{K}}_2=
\left[\begin{array}{cc|cc}
\lambda_2 & 0 & 0 & 0 \\ 
0 & \lambda_1 & 0 & 0\\\hline 
0 & 0 & \lambda_2 & 0 \\
0 & 0 & 0 & \lambda_1 \ .
\end{array}\right]
$$ 
(To be more precise, in order to have a Darboux--Haantjes chart one should eliminate the factor $i$ in the form of $\check{P_1}$ above, by means of the map $\lambda\mapsto i \lambda$, $\mu\mapsto \mu$.)
\par
As a consequence of  Theorem  59 in \cite{TT2016prepr}, the coordinates $(\lambda_1,\lambda_2,\mu_1,\mu_2)$ are separation variables for the functions  $F_2$ and $F_3$.  Furthermore, they are  also separation variables for the (restriction to $S_1$ of the)  Hamiltonian function of the Lagrange top 
\begin{equation}\label{eq:HLrid}
h_{1{|S_1}}=-F_{3{|S_1}}-(c-1)C_1 F_{2_{|S_1}} \ .
\end{equation}
This fact can be proved by means of the Benenti test \cite{Ben80}, or simply  by observing that, thanks to Eq. \eqref{eq:MHcF3F2}, it holds true that 
$$
\rd h_{1{|S_1}}=(-\boldsymbol{\check{I}}+(c-1)C_1 \boldsymbol{\check{K}}_2^T\big)\rd F_{3{|S_1}}\ .
$$
Therefore the  function $h_{1{|S_1}}$ belongs to a Magri--Haantjes chain generated by $F_{3{|S_1}}$. Consequently, according to    Theorem  59 in \cite{TT2016prepr}, also $h_{1{|S_1}}$ is separable in any $DH$ local chart.

\section{Future Perspectives}
It would be interesting to construct a Poisson--Haantjes algebra for  the Poisson pencil $P_2-\lambda P_1$ of the Lagrange top, that unlike  the pencil $P_0-\lambda P_1$,  has two polynomial Casimir functions  of the same length 2.  Moreover, a Haantjes algebra for  the stationary flows \cite{FMT} of the Boussinesq hierarchy, that also are GZ systems of corank 2, should be worked out.

\section*{Acknowledgement}
\noi The author  is a  member of the Gruppo Nazionale di Fisica Matematica (GNFM) of INDAM.
\par
He wishes to thank many useful discussions with his colleagues and friends S.~Logar and D. Portelli. Also, many thanks are due to an anonymous referee for the careful reading of the manuscript and for valuable comments and stimulating questions.

\end{document}